\documentclass[11pt]{article}
\usepackage[dvipsnames,svgnames,x11names,hyperref]{xcolor}
\usepackage{fullpage}
\usepackage{amsmath}
\usepackage{amsthm}
\usepackage{thmtools}
\usepackage[colorlinks]{hyperref}
\usepackage{fullpage,amssymb,amsfonts,amsmath,amsthm,amscd,dsfont,mathrsfs,bbm}
\usepackage{mathpazo}
\usepackage{todonotes}
\usepackage{mdframed}
\usepackage{tikz}
\usetikzlibrary{decorations.pathreplacing,backgrounds}
\usepackage{multirow}

\usepackage{algorithmicx}
\usepackage{algorithm} 
\usepackage[noend]{algpseudocode}

\algrenewcommand\algorithmicindent{1.0em}%

\usepackage{scrextend} 
\newcommand*\samethanks[1][\value{footnote}]{\footnotemark[#1]}

\usetikzlibrary{arrows}

\declaretheoremstyle[bodyfont=\it,qed=$\qed$]{noproofstyle}

\declaretheorem[numberlike=equation]{observation}

\declaretheorem[numberlike=equation]{subclaim}

\declaretheorem[numberlike=equation]{theorem}

\declaretheorem[numberlike=equation]{lemma}

\declaretheorem[numberlike=equation]{corollary}

\declaretheorem[numberlike=equation]{claim}

\declaretheorem[numberlike=equation]{question}

\declaretheorem[unnumbered,name=Theorem,style=noproofstyle]{theorem*}
\declaretheorem[unnumbered,name=Lemma,style=noproofstyle]{lemma*}
\declaretheorem[unnumbered,name=Corollary,style=noproofstyle]{corollary*}
\declaretheorem[unnumbered,name=Proposition,style=noproofstyle]{proposition*}

\declaretheorem[unnumbered,name=Claim]{claim*}
\declaretheorem[unnumbered,name=Conjecture]{conjecture*}
\declaretheorem[unnumbered,name=Question]{question*}

\declaretheoremstyle[bodyfont=\it]{defstyle} 
\declaretheorem[numberlike=equation,style=defstyle]{definition}
\declaretheorem[unnumbered,name=Definition,style=defstyle]{definition*}

\declaretheorem[unnumbered,name=Example,style=defstyle]{example*}

\declaretheorem[unnumbered,name=Notation=defstyle]{notation*}

\declaretheorem[unnumbered,name=Construction,style=defstyle]{construction*}

\declaretheorem[unnumbered,name=Remark,style=defstyle]{remark*}

\newenvironment{myproof}[1]%
{\vspace{1ex}\noindent{\emph{Proof.}}\hspace{0.5em}\def\myproof@name{#1}\newif\ifqedhere\qedherefalse}%
{\ifqedhere \else \hfill{\tiny \qed\ (\myproof@name)}\vspace{1ex} \qedherefalse \fi}

\newcommand{\myqedhere}
{\global\qedheretrue \tag*{\text{\tiny \qed\ (\myproof@name)}}}

\newenvironment{proof-sketch}{\medskip\noindent{\em Sketch of Proof.}\hspace*{1em}}{\qed\bigskip}
\newenvironment{proof-attempt}{\medskip\noindent{\em Proof attempt.}\hspace*{1em}}{\bigskip}

\newcommand{\inparen }[1]{\left(#1\right)}             
\newcommand{\inbrace }[1]{\left\{#1\right\}}           

\newcommand{\eqdef}{\stackrel{\mathrm{def}}{=}}
\newcommand{\setdef}[2]{\inbrace{{#1}\ \mid \ {#2}}}      

\newcommand{\F}{\mathbb{F}}

\newcommand{\Monset}{\mathbb{M}}

\newcommand{\poly}{\mathrm{poly}}

\newcommand{\veca}{\mathbf{a}}
\newcommand{\vecb}{\mathbf{b}}
\newcommand{\vecc}{\mathbf{c}}

\newcommand{\vece}{\mathbf{e}}

\newcommand{\vecu}{\mathbf{u}}
\newcommand{\vecv}{\mathbf{v}}

\newcommand{\vecx}{\mathbf{x}}
\newcommand{\vecy}{\mathbf{y}}

\newcommand{\spaced}[1]{\quad #1 \quad}
\newcommand{\sspaced}[1]{\; #1 \;} 

\renewcommand{\epsilon}{\varepsilon}

\usepackage{blkarray}

\newcommand\independent{\protect\mathpalette{\protect\independent}{\perp}} 
\def\independent#1#2{\mathrel{\rlap{$#1#2$}\mkern2mu{#1#2}}}

\newcommand{\e}{\varepsilon}

\newcommand{\1}{\mathbbm{1}}

\theoremstyle{plain}
\newtheorem{defin}[theorem]{Definition}
\theoremstyle{plain}
\theoremstyle{plain}

\theoremstyle{plain}
\theoremstyle{discussion}
\theoremstyle{plain}

\theoremstyle{plain}

\usepackage{nth}
\usepackage{intcalc}
\usepackage{etoolbox}
\usepackage{xstring}

\newcommand{\ECCCurl}[2]{http://eccc.hpi-web.de/report/\ifnumcomp{#1}{>}{93}{19}{20}#1/#2/}

\newcommand{\shortECCC}[2]{\texttt{\href{http://eccc.hpi-web.de/report/\ifnumcomp{#1}{>}{93}{19}{20}#1/#2/}{eccc:TR#1-#2}}}

\newcommand{\parseECCC}[1]{
\StrSubstitute{#1}{TR}{}[\tmpstring]%
\IfSubStr{\tmpstring}{/}{ 
\StrBefore{\tmpstring}{/}[\ecccyear]%
\StrBehind{\tmpstring}{/}[\ecccreport]%
}{
\StrBefore{\tmpstring}{-}[\ecccyear]%
\StrBehind{\tmpstring}{-}[\ecccreport]%
}%
\shortECCC{\ecccyear}{\ecccreport}}

\hypersetup{
  linkcolor=[rgb]{0.3,0.3,0.6},
  citecolor=[rgb]{0.2, 0.6, 0.2},
  urlcolor=[rgb]{0.6, 0.2, 0.2}
}

\begin{document}
\title{Efficiently decoding Reed-Muller codes from random errors}
\author{Ramprasad Saptharishi\thanks{Department of Computer Science, Tel Aviv University, Tel Aviv, Israel, E-mails: \texttt{ramprasad@cmi.ac.in, benleevolk@gmail.com}. The research leading to these results has received funding from the European Community's Seventh Framework Programme (FP7/2007-2013) under grant agreement number 257575.}%
\and%
Amir Shpilka\thanks{Department of Computer Science, Tel Aviv University, Tel Aviv, Israel,
\texttt{shpilka@post.tau.ac.il}.  The research leading to these results has received funding
from the European Community's Seventh Framework Programme (FP7/2007-2013) under grant agreement number 257575, and from the Israel Science Foundation (grant number 339/10).}
\and%
Ben Lee Volk\samethanks[1]
}
\date{}
\maketitle

\begin{abstract}
Reed-Muller codes encode an $m$-variate polynomial of degree $r$ by evaluating it on all points in $\{0,1\}^m$.
We denote this code by $RM(m,r)$. The minimal distance of $RM(m,r)$ is $2^{m-r}$ and so it cannot correct more than half that number of errors in the worst case. For random errors one may hope for a better result.

In this work we give an efficient algorithm (in the block length $n=2^m$) for decoding random errors in Reed-Muller codes far beyond the minimal distance.
Specifically, for low rate codes (of degree $r=o(\sqrt{m})$) we can correct a random set of $(1/2-o(1))n$ errors with high probability.
For high rate codes (of degree $m-r$ for $r=o(\sqrt{m/\log m})$), we can correct roughly $m^{r/2}$ errors.

More generally, for any integer $r$, our algorithm can correct any error pattern in $RM(m,m-(2r+2))$ for which the same erasure pattern can be corrected in $RM(m,m-(r+1))$.
The results above are obtained by applying recent results of Abbe, Shpilka and Wigderson (STOC, 2015), Kumar and Pfister (2015) and Kudekar et al.\ (2015) regarding the ability of Reed-Muller codes to correct random erasures.
 
The algorithm is based on solving a carefully defined set of linear equations and thus it is significantly different than other algorithms for decoding Reed-Muller codes that are based on the recursive structure of the code.
It can be seen as a more explicit proof of a result of Abbe et al.\ that shows a reduction from correcting erasures to correcting errors, and it also bares some similarities with the famous Berlekamp-Welch algorithm for decoding Reed-Solomon codes.

\end{abstract}
\thispagestyle{empty}
\newpage

\pagenumbering{arabic}

\section{Introduction}

Consider the following challenge:
\begin{quote}
  Given the truth table of a polynomial $f(\vecx) \in \F_2[x_1,\dots, x_m]$ of degree at most $r$, in which $1/2-o(1)$ fraction of the locations were flipped (that is, given the evaluations of $f$ over $\F_2^m$ with nearly half the entries corrupted), recover $f$ efficiently.
\end{quote}

\noindent 
If the errors are adversarial, then clearly this task is impossible for any degree bound $r \ge 2$, since there are two different quadratic polynomials that disagree on only $1/4$ fraction of the domain.
Hence, we turn to considering {\em random} sets of errors of size $(1/2-o(1))2^m$, and we hope to recover $f$ with high probability (in this case, one may also consider the setting where each bit is independently flipped with probability $1/2-o(1)$.
By standard Chernoff bounds, both settings are almost equivalent).

Even in the random model, if every bit was flipped with probability exactly $1/2$, the situation is again hopeless: in this case the input is completely random and carries no information whatsoever about the original polynomial.

It turns out, however, that even a very small relaxation leads to a dramatic improvement in our ability to recover the hidden polynomial: in this paper we prove, among other results, that even at corruption rate $1/2-o(1)$ and degree bound as large as $o(\sqrt{m})$, we can {\em efficiently} recover the {\em unique} polynomial $f$ whose evaluations were corrupted.
Note that in the worst case, given a polynomial of such a high degree, an adversary can flip a tiny fraction of the bits --- just slightly more than $1/2^{\sqrt{m}}$ --- and prevent unique recovery of $f$, even if we do not require an efficient solution; and yet, in the average case, we can deal with flipping almost half the bits.

Recasting the playful scenario above in a more traditional terminology, this paper deals with similar questions related to recovery of low-degree multivariate polynomials from their \emph{randomly} corrupted evaluations on $\F_2^m$, or in the language of coding theory, we study the problem of decoding \emph{Reed-Muller} codes under random errors in the \emph{binary symmetric channel (BSC)}.
We turn to some background and motivation.

\subsection{Reed-Muller Codes}

Reed-Muller (RM) codes were introduced in 1954, first by Muller \cite{muller} and shortly after by Reed \cite{reed} who also provided a decoding algorithm.
They are among the oldest and simplest codes to construct --- the codewords are multivariate polynomials of a given degree, and the encoding function is just their evaluation vectors.
In this work we mainly focus on the most basic case where the underlying field is $\F = \F_2$, the field of two elements, although our techniques do generalize to larger finite fields.
Over $\F_2$, the Reed-Muller code of degree $r$ in $m$ variables, denoted by $RM(m,r)$, has block length $n=2^m$, rate $\binom{m}{\le r}/2^m$ and its minimal distance is $2^{m-r}$.

RM codes have been extensively studied with respect to decoding errors in both the worst case and random setting.
We begin by giving a review of Reed-Muller codes and their use in theoretical computer science and then discuss our results.

\subsubsection*{Background}

Error-correcting codes (over both large and small finite fields) have been extremely influential in the theory of computation, playing a central role in some important developments in several areas such as cryptography (e.g.\ \cite{Shamir79} and \cite{BF90}), theory of pseudorandomness (e.g.\ \cite{bogdanov-viola}), probabilistic proof systems (e.g.\ \cite{BFL91,Sha92} and \cite{ALMSS98}) and many more.

An important aspect of error correcting codes that received a lot of attention is designing efficient decoding algorithms.
The objective is to come up with an algorithm that can correct a certain amounts of errors in a received word.
There are two settings in which this problem is studied:

\medskip

{\bf Worst case errors:} This is also referred to as errors in the \emph{Hamming model}~\cite{hamming50}.
Here, the algorithm should recover the original message regardless of the error pattern, as long as there are not too many errors.
The number of errors such a decoding algorithm can tolerate is upper bounded in terms of the distance of the code.
The {\em distance} of the code $C$ is the minimum Hamming distance of any two codewords in $C$.
If the distance is $d$, then one can {\em uniquely} recover from at most $d-1$ erasures and from $\lfloor (d-1)/2 \rfloor$ errors.
For this model of worst-case errors it is easy to prove that Reed-Muller codes perform badly.
They have relatively small distance compared to what random codes of the same rate can achieve (and also compared to explicit families of codes).

Another line of work in Hamming's worst case setting concerns designing algorithms that can correct beyond the unique-decoding bound.
Here there is no unique answer and so the algorithm returns a list of candidate codewords.
In this case the number of errors that the algorithm can tolerate is a parameter of the distance of the code.
This question received a lot of attention and among the works in this area we mention the seminal works of Goldreich and Levin on Hadamard Codes \cite{GoldreichLevin89} and of Sudan \cite{Sudan97} and Guruswami and Sudan \cite{GuruswamiSudan99} on list decoding Reed-Solomon codes.
Recently, the list-decoding question for Reed-Muller codes was studied by Gopalan, Klivans and Zuckerman \cite{GopalanKZ08} and by Bhowmick and Lovett \cite{BhowmickL14}, who proved that the list decoding radius\footnote{The maximum distance $\eta$ for which the number of code words within distance $\eta$ is only polynomially large (in $n$).}
of Reed-Muller codes, over $\F_2$, is at least twice the minimum distance (recall that the unique decoding radius is half that quantity) and is smaller than four times the minimal distance, when the degree of the code is constant.

\medskip

{\bf Random errors:} A different setting in which decoding algorithms are studied is Shannon's model of random errors \cite{shannon48}.
In Shannon's average-case setting (which we study here), a codeword is subjected to a random corruption, from which recovery should be possible {\em with high probability}.
This random corruption model is called a {\em channel}.
The two most basic ones, the Binary Erasure Channel (BEC) and the Binary Symmetric Channel (BSC), have a parameter $p$ (which may depend on $n$), and corrupt a message by independently replacing, with probability $p$, the symbol in each coordinate, with a ``lost'' symbol in the BEC($p$) channel, and with the complementary symbol in the BSC($p$) case.
In his paper Shannon studied the optimal trade-off achievable for these channels (and many other channels) between the distance and rate.
For {\em every} $p$, the capacity of BEC($p$) is $1-p$, and the capacity of BSC($p$) is $1-h(p)$, where $h$ is the binary entropy function.\footnote{$h(p) = -p\log_2(p) - (1-p)\log_2(1-p)$, for $p\in (0,1)$, and $h(0)=h(1)=0$.}
Shannon also proved that random codes achieve this optimal behavior. That is, for every $0<\epsilon$ there exist codes of rate $1-h(p)-\epsilon$ for the BSC (and rate $1-p-\epsilon$ for the BEC), that can decode from a fraction $p$ of errors (erasures) with high probability.

For our purposes, it is more convenient to assume that the codeword is subjected to a fixed number $s$ of random errors.
Note that by the Chernoff-Hoeffding bound, (see e.g., \cite{AlonSpencer}), the probability that more than $pn + \omega(\sqrt{pn})$ errors occur in BSC($p$) (or BEC($p$)) is $o(1)$, and so we can restrict ourselves to the case of a fixed number $s$ of random errors, by setting the corruption probability to be $p=s/n$.
We refer to \cite{AbbeSW15} for further discussion on this subject.

\subsubsection*{Decoding erasures to decoding errors} \label{sec:erasures}

Recently, there has been a considerable progress in our understanding of the behavior of Reed-Muller codes under random erasures.
In \cite{AbbeSW15}, Abbe, Shpilka and Wigderson showed that Reed-Muller codes achieve capacity for the BEC for both sufficiently low and sufficiently high rates.
Specifically, they showed that $RM(m,r)$ achieves capacity for the BEC for $r = o(m)$ or $r > m - o(\sqrt{m/\log m})$.
More recently, Kumar and Pfister~\cite{KumarPfister15} and Kudekar, Mondelli, \c{S}a\c{s}o\u{g}lu and Urbanke \cite{KudekarMSU15} independently showed that Reed-Muller codes achieve capacity for the BEC in the entire constant rate regime, that is $r \in [m/2-O(\sqrt{m}), m/2+O(\sqrt{m})]$.
These regimes are pictorially represented in \autoref{fig:rm-bec-regime}.

\begin{figure}[h]
\begin{center}
\begin{tikzpicture}[framed]
\draw[thick] (0,0) -- (10,0);
\draw (0,-0.25) -- (0,0.25);
\draw (10,-0.25) -- (10,0.25);
\draw (5,-0.25) -- (5,0.25);
\node at (5,0.5) {$m/2$};
\node at (0,0.5) {$0$};
\node at (10,0.5) {$m$};
\draw[ultra thick,draw=red] (4,0) -- (6,0);
\draw[ultra thick,draw=red] (0,0) -- (2,0);
\draw[ultra thick,draw=red] (8.5,0) -- (10,0);
\draw [decorate,decoration={brace,amplitude=10pt,mirror},yshift=-3pt]
(0,0) -- (2,0) node [black,midway,yshift=-0.8cm]
{\footnotesize $o(m)$};

\draw [decorate,decoration={brace,amplitude=10pt,mirror},yshift=-3pt]
(8.5,0) -- (10,0) node [black,midway,yshift=-0.8cm]
{\footnotesize $o(\sqrt{(m/\log m)})$};
\draw [decorate,decoration={brace,amplitude=10pt,mirror},yshift=-3pt]
(4,0) -- (6,0) node [black,midway,yshift=-0.8cm]
{\footnotesize $O(\sqrt{m})$};
\end{tikzpicture}
\end{center}
\caption{Regime of $r$ for which $RM(m,r)$ is known to achieve capacity for the BEC}
\label{fig:rm-bec-regime}
\end{figure}
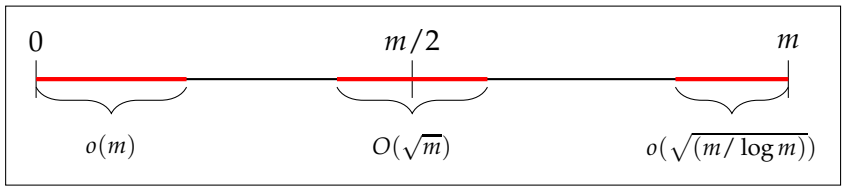

Another result proved by Abbe et al.\ \cite{AbbeSW15} is that Reed-Muller codes $RM(m,m-2r-2)$ can correct any \emph{error pattern} if the same \emph{erasure pattern} can be decoded in $RM(m,m-r-1)$.
This reduction is appealing on its own, since it connects decoding from erasures --- which is easier in both an intuitive and an algorithmic manner --- with decoding from errors; but its importance is further emphasized by the progress made later by Kumar and Pfister and Kudekar et al., who showed that Reed-Muller codes can correct many erasures in the constant rate regime, right up to the channel capacity.

This result show that $RM(m,m-(2r+2))$ can cope with most error patterns of weight $(1-o(1))\binom{m}{\le r}$, which is the capacity of $RM(m,m-(r+1))$ for the BEC.
While this is polynomially smaller than what can be achieved in the Shannon model of errors for random codes of the same rate, this number is still much larger (super-polynomial) than the distance (and the list-decoding radius) of the code, which is $2^{2r+2}$.
Also, since $RM\inparen{m, \frac{m}{2} - o(\sqrt{m})}$ can cope with $\inparen{\frac{1}{2} - o(1)}$-fraction of erasures, this translation implies that $RM(m,o(\sqrt{m}))$ can handle that many random errors.

However, a shortcoming of the proof of Abbe et al.\ for the BSC is that it is existential.
In particular it does not provide an efficient decoding algorithm.
Thus, Abbe et al.\ left open the question of coming up with a decoding algorithm for Reed-Muller codes from random errors.

\subsection{Our contributions}\label{sec:our-results}

In this work we give an efficient decoding algorithm for Reed-Muller codes that matches the parameters given by Abbe et al.
Following the aforementioned results about the erasure correcting ability of Reed-Muller codes, the results can be partitioned into the low-rate and the high-rate regimes.
We begin with the result for the low rate case.

\begin{theorem}[Low rate, informal]\label{thm:main-low-degree}
  Let $r < \delta \sqrt{m}$ for a small enough $\delta$.
Then, there is an efficient algorithm that can decode $RM(m,r)$ from a random set of $(1 - o(1)) \cdot \binom{m}{\leq m/2 - r}$ errors.
In particular, if $r = o(\sqrt{m})$, the algorithm can decode from $\inparen{\frac{1}{2} - o(1)} \cdot 2^m$ errors.
The running time of the algorithm is $O(n^4)$ and it can be simulated in $\mathsf{NC}$.
\end{theorem}

For high rate Reed-Muller codes, we cannot hope to achieve such a high error correction capability as in the low rate case, even information theoretically.
We do give, however, an algorithm that corrects many more errors (a super-polynomially larger number) than what the minimal distance of the code suggests, and its running time is also nearly linear in the block length of the code.

\begin{theorem}[High rate, informal]\label{thm:main:informal}
  Let $r =o(\sqrt{m/\log m})$.
Then, there is an efficient algorithm that can decode $RM(m, m-(2r+2))$ from a random set of $(1-o(1))\binom{m}{\le r}$ errors.
Moreover, the running time of the algorithm is $2^m \cdot \poly(\binom{m}{\le r})$ and it can be simulated in $\mathsf{NC}$.
\end{theorem}

Recall that the block length of the code is $n=2^m$, and thus the running time is near linear in $n$ when $r=o(m)$.

A general property of our algorithm is that it corrects any error pattern in $RM(m, m - 2r -2)$ for which the same {\em erasure} pattern in $RM(m,m - r-1)$ can be corrected.
Stated differently, if an erasure pattern can be corrected in $RM(m, m-r-1)$ then the same pattern, where the ``lost'' symbol is replaced with arbitrary $0/1$ values, can be corrected in $RM(m, m-(2r+2))$.
This property is useful when we know $RM(m,m - r-1)$ can correct a large set of erasures with high probability, that is, when $m-r-1$ falls in the \emph{red region} in \autoref{fig:rm-bec-regime}.
Thus, our result has implications also beyond the above two instances.
In particular, it may be the case that our algorithm performs well for other rates as well.
For example, consider the following question and the theorem it implies.

\begin{question}\label{Q:RM:BEC}
Does $RM(m, m-r-1)$ achieve capacity for the BEC?
\end{question}

\begin{theorem}[informal]\label{thm:main:conj:informal}
  For any value $r$ for which the answer to \autoref{Q:RM:BEC} is positive, there exists an efficient algorithm that decodes $RM(m, m-2r-2)$ from a random set of $(1-o(1))\binom{m}{\le r}$ errors with probability $(1-o(1))$ (over the random errors).
Moreover, the running time of the algorithm is $2^m \cdot \poly\inparen{\binom{m}{\le r}}$.
\end{theorem}


Recall that Abbe et al.\ \cite{AbbeSW15} also proved that the answer to \autoref{Q:RM:BEC} is positive for $r = m- o(m)$ (that is, for $RM(m,o(m))$) but this case does not help us as we need to consider $RM(m, m-(2r+2))$ and $m - (2r + 2) < 0$ in this case.
The coding theory community seems to believe the answer to \autoref{Q:RM:BEC} is positive, for all values of $r$, and conjectures to that effect were made\footnote{The belief that RM codes achieve capacity is much older, but we did not trace back where it appears first.}
in \cite{forney-road,arikan-RM,mondelli-RM}.
Recent simulations have also suggested that the answer to the question is positive \cite{arikan-RM,mondelli-RM}.
Thus, it seems natural to believe that the answer is positive for most values of $r$, even for $r =\Theta(m)$.
As a conclusion, the belief in the coding theory community suggests that our algorithm can decode a random set of roughly $\binom{m}{\le r}$ errors in $RM(m, m-(2r+2))$.
For example, for $r=\rho\cdot m$, where $\rho<1/2$, the minimal distance of $RM(m, m-(2r+2))$ is roughly $2^{2\rho m}$ whereas our algorithm can decode from roughly $2^{h(\rho)m}$ random errors (assuming the answer to \autoref{Q:RM:BEC} is positive), which is a much larger quantity for every $\rho < 1/2$.\\

In \autoref{sec:abstraction}, we also present an abstraction of our decoding procedure that may be applicable to other linear codes.
This is a generalization of the abstract Berlekamp-Welsch decoder or ``error-locating pairs'' method of Duursma and K\"{o}tter~\cite{DuursmaK94} that connects decodable erasure patterns on a larger code to decodable error patterns.
A specific instantiation of this was observed by Abbe et al.\ \cite{AbbeSW15} by connecting decodable error patterns of any linear code $C$ to decodable erasure patterns of an appropriate ``tensor'' $C'$ of $C$ (by essentially embedding these codes in a large enough RM code).
Although Abbe et al.\ did not provide an efficient decoding algorithm, the algorithm we present directly applies here (\autoref{sec:general}).
The abstraction of the ``error-locating pairs'' method presented in \autoref{sec:abstraction} should hopefully be applicable in other contexts too, especially considering the generality of the results of \cite{KumarPfister15, KudekarMSU15}.

\subsection{Related literature} 

In \autoref{sec:erasures} we surveyed the known results regarding the ability of Reed-Muller codes to correct random erasures.
In this section we summarize the results known about recovering RM codes from random errors.

Once again, it is useful to distinguish between the low rate and the high rate regime of Reed-Muller codes.
We shall use $d$ to denote the distance of the code in context.
For $RM(m,r)$ codes, $d = 2^{m-r}$.

In \cite{krich}, the majority logic algorithm of \cite{reed} is shown to succeed in recovering all but a vanishing fraction of error patterns of weight up to $d \log d/4$ for all RM codes of positive rate.
In \cite{dumer3}, Dumer showed for all $r$ such that $\min(r, m-r) = \omega(\log m)$ that most error patterns of weight at most $(d\log d/2) \cdot (1 - \frac{\log m}{\log d})$ can be recovered in $RM(m,r)$.
To make sense of the parameters, we note that when $r = m - \omega(\log m)$ 
the weight is roughly $(d\log d/2)$.
To compare this result to ours, we first consider the case when $r=m - o(\sqrt{m/\log m})$.
Here the algorithm of \cite{dumer3} can correct roughly $2^{o(\sqrt{m/\log m})}$ random errors in $RM(m,r)$ whereas \autoref{thm:main:informal} gives an algorithm for correcting roughly $m^{o(\sqrt{m/ \log m})} \approx (d \log d)^{O(\log m)}$ random errors.

Further, even for the case $r = (1 -\rho) m$, where $\rho<1/2$ is a constant, the bound in the above result of \cite{dumer3} is equal to $O(d\log d)$.
On the other hand, assuming a positive answer to \autoref{Q:RM:BEC}, \autoref{thm:main:conj:informal} implies an efficient decoding algorithm for $RM(m,(1-\rho)m)$ that can decode from, roughly, $\binom{m}{\frac{1}{2}\rho m} = d^{O(\log 1/\rho)}$ random errors, for this case.


\begin{figure}[h]
\begin{center}
\begin{tikzpicture}[framed]
\draw[thick] (0,0) -- (12,0);
\draw (0,-0.25) -- (0,0.25);
\draw (12,-0.25) -- (12,0.25);
\draw (6,-0.25) -- (6,0.25);
\node at (6,-0.5) {\scriptsize $m/2$};
\node at (0,-0.5) {\scriptsize $0$};
\node at (12,-0.5) {\scriptsize $m$};

\draw [decorate,decoration={brace,amplitude=5pt},yshift=10pt]
(0,0) -- (1,0) node [black,midway,yshift=0.4cm]
{\scriptsize $\log m$};

\draw [decorate,decoration={brace,amplitude=5pt},yshift=10pt]
(11,0) -- (12,0) node [black,midway,yshift=0.4cm]
{\scriptsize $\log m$};
\draw (1,-0.25) -- (1,0.25);
\draw (11,-0.25) -- (11,0.25);

\draw (2,-0.25) -- (2,0.25);
\node at (2,-0.5) {\scriptsize $o(\sqrt{m})$};


\draw (9,-0.25) -- (9,0.25);
\draw [decorate,decoration={brace,amplitude=10pt,mirror},yshift=-20pt]
(9,0) -- (12,0) node [black,midway,yshift=-0.7cm]
{\scriptsize $o(\sqrt{m/\log m})$};

\node[anchor=east] at (-0.5,0) {\scriptsize Degree ($r$) of $RM(m,r)$:};


\node[anchor=east] at (-0.5,1.5) {\scriptsize \cite{dumer1,dumer2,dumer3}:};

\draw[very thick, draw=blue!80] (0,1.5) -- (0.9,1.5) node [black, midway, above] {\scriptsize $\approx n/2$ errors};

\draw[very thick, draw=red!80] (1.1,1.5) -- (10.9,1.5) node [black, midway, above] {\scriptsize $O(d \log d)$ errors} node [black, midway, below] {\scriptsize $O(n \log n)$ time algorithm};


\node[anchor=east] at (-0.5,3) {\scriptsize Our results:};

\draw[very thick, draw=blue!80] (0,3) -- (2,3) node [black, midway, above] {\scriptsize $\approx n/2$ errors} node [black, midway, below] {\scriptsize $O(n^4)$ time algo.};

\draw[very thick, draw=red!80] (9,3) -- (12,3) node [black, midway, above] {\scriptsize $(d \log d)^{O(\log m)}$ errors} node [black, midway, below] {\scriptsize $n^{1 + o(1)}$ time algo.};

\draw[very thick, loosely dotted, draw=brown!80] (2,3) -- (9,3) node [black!70, midway, above] {\scriptsize $(d \log d)^{\omega(1)}$ errors} node [black!70, midway, below] {\scriptsize assuming positive answer to \autoref{Q:RM:BEC}};
\end{tikzpicture}
\end{center}
\caption{Comparison with \cite{dumer1,dumer2,dumer3}}
\label{fig:comparison-dumer}
\end{figure}
We now turn to other regimes of parameters, specifically RM codes of low rate.
For the special case of $r=1,2$, \cite{hell} shows that $RM(m,r)$ codes are capacity-achieving.
In \cite{sidel}, it is shown that RM codes of fixed order (i.e., $r=O(1)$) can decode most error patterns of weight up to $\frac{1}{2}n(1-\sqrt{c(2^r-1)m^r/ n r!})$, where $c> \ln(4)$.
In \cite{AbbeSW15}, Abbe et al.\ settled the question for low order Reed-Muller codes proving that $RM(m,r)$ codes achieve capacity for the BSC when $r=o(m)$ \cite{AbbeSW15}.
We note however that all the results mentioned here are existential in nature and do not provide an efficient decoding algorithm.

A line of work by Dumer \cite{dumer1,dumer2} based on recursive algorithms (that exploit the recursive structure of Reed-Muller codes), obtains algorithmic results mainly for low-rate regimes.
In \cite{dumer1}, it is shown that for a fixed degree, i.e., $r=O(1)$, an algorithm of complexity $O(n \log n)$ can correct most error patterns of weight up to $n(1/2 - \e)$ given that $\e$ exceeds $n^{-1/2^r}$.
In \cite{dumer3}, this is improved to errors of weight up to $\frac{1}{2}n(1-(4m/d)^{1/2^r})$ for all $r=o(\log m )$.
The case $r=\omega(\log m)$ is also covered in \cite{dumer3}, as described above.

We note that all the efficient algorithms mentioned above (both for high- and low-rate) rely on the so called Plotkin construction of the code, that is, on its recursive structure (expanding an $m$-variate polynomial according to the $m$-th variable $f(x_1,\ldots,x_m)=x_m g(x_1,\ldots,x_{m-1})+h(x_1,\ldots,x_{m-1})$), whereas our approach is very different.

We summarize and compare our results with \cite{dumer1,dumer2,dumer3} for various range of parameters in \autoref{fig:comparison-dumer} (degree is $r$ and distance is $d = 2^{m-r}$).
The dotted region in \autoref{fig:comparison-dumer} corresponds to the uncovered region in \autoref{fig:rm-bec-regime} beyond $m/2$, via  the connection given in \autoref{thm:main:conj:informal}.


\subsection{Notation and terminology}\label{sec:notation}
Before explaining the idea behind the proofs of our results we need to introduce some notation and parameters.
We shall use the same notation as \cite{AbbeSW15}.

\begin{itemize}
\item We denote by $\Monset(m,r)$ the set of $m$-variate monomials over $\F_2$ of degree at most $r$.

\item For non-negative integers $r \leq m$, $RM(m,r)$ denotes the
  Reed-Muller code whose codewords are the evaluation vectors of all multivariate polynomials of degree at most $r$ on $m$ boolean variables.
The maximal degree $r$ is sometimes called the order of the code.
The block length of the code is $n=2^m$, the dimension $k=k(m,r)=\sum_{i=0}^r \binom{m}{i} \eqdef \binom{m}{\le r}$, and the distance $d=d(m,r)=2^{m-r}$.
The code rate is given by $R=k(m,r)/n$.

\item We use $E(m,r)$ to denote the ``evaluation matrix'' of parameters $m,r$, whose rows are indexed by all monomials in  $\Monset(m,r)$, and whose columns are indexed by all vectors in $\F_2^m$. The value at entry $(M,\vecu)$ is equal to $M(\vecu)$. 
For $\vecu\in \F_2^m$, we denote by $\vecu^r$ the column of $E(m,r)$ indexed by $\vecu$, which is a $k$-dimensional vector, consisting of all evaluations of degree $\leq r$ monomials at $\vecu$. For a subset of columns $U \subseteq \F_2^m$ we denote by $U^r$ the corresponding submatrix of $E(m,r)$.

\item $E(m,r)$ is a generator matrix for $RM(m,r)$. The duality
  property of Reed-Muller codes (see, for example, \cite{sloane-book}) states that $E(m,m-r-1)$ is a parity-check matrix for $RM(m,r)$, or equivalently, $E(m,r)$ is a parity-check matrix for $RM(m,m-r-1)$.

\item We associate with a subset $U\subseteq \F_2^m$ its
  characteristic vector $\1_U \in \F_2^n$.
We often think of the vector $\1_U$ as denoting either an {\em erasure pattern} or an {\em error pattern}.

\item For a positive integer $n$, we use the standard notation $[n]$ for the set $\{1,2,\ldots,n\}$.

\end{itemize}

We next define what we call the degree-$r$ syndrome of a set.
\begin{definition}[Syndrome]
  Let $r \le m$ be two positive integers.
The \emph{degree-$r$ syndrome}, or simply \emph{$r$-syndrome} of a set $U=\inbrace{\vecu_1,\ldots,\vecu_t}\subseteq\F_2^m$ is the $\binom{m}{\le r}$-dimensional vector $\alpha$ whose entries are indexed by all monomials $M \in \Monset(m,r)$, such that
\[
\alpha_M \eqdef \sum_{i=1}^t M(\vecu_i).
\]
\end{definition}
Note that this is nothing but the syndrome of the error pattern $\1_U \in \F_2^{n}$ in the code $RM(m, m-r-1)$ (whose parity check matrix is the generator matrix of $RM(m,r)$).


\subsection{Proof techniques}\label{sec:techniques}

In this section we describe our approach for constructing a decoding algorithm.
Recall that the algorithm has the property that is decodes in $RM(m, m-2r-2)$ any error pattern $U$ which is correctable from erasures in $RM(m, m-r-1)$.
Such patterns are characterized by the property that the columns of $E(m,r)$ corresponding to the elements of $U$ are linearly independent vectors.
Thus, it suffices to give an algorithm that succeeds whenever the error pattern $\1_U$ gives rise to such linearly independent columns, which happens with probability $1-o(1)$ for the regime of parameters mentioned in \autoref{thm:main-low-degree} and \autoref{thm:main:informal}.

So let us assume from now on that the error pattern $\1_U$ corresponds to a set of linearly independent columns in $E(m,r)$.
Notice that by the choice of our parameters, our task is to recover $U$ from the degree $(2r+1)$-syndrome of $U$.
Furthermore, we want to do so efficiently.
For convenience, let $t = |U|=(1-o(1))\binom{m}{\le r}$.

Recall that the degree-$(2r+1)$ syndrome of $U$ is the $\binom{m}{\le 2r+1}$-long vector $\alpha$ such that for every monomial $M\in \Monset(m,2r+1)$, $\alpha_M = \sum_{i=1}^t M(\vecu_i)$.
Imagine now that we could somehow find degree-$r$ polynomials $f_i(x_1,\ldots,x_m)$ satisfying $f_i(u_j)=\delta_{i,j}$.
Then, from knowledge of $\alpha$ and, say, $f_1$, we could compute the following sums:
\[
\sigma_\ell = \sum_{i=1}^{t} (f_1 \cdot x_\ell) (\vecu_i), \quad \ell \in [m].
\]
Indeed, if we know $\alpha$ and $f_1$ then we can compute each $\sigma_\ell$, as it just involves summing several coordinates of $\alpha$ (since $\deg(f_1\cdot x_\ell) \leq r+1$).
We now observe that
\[
\sigma_\ell = \sum_{i=1}^{t} (f_1 \cdot x_\ell) (\vecu_i) = (f_1 \cdot x_\ell) (\vecu_1) = ( \vecu_1)_\ell.
\]
In other words, knowledge of such an $f_1$ would allow us to discover all coordinates of $\vecu_1$ and in particular, we will be able to deduce $\vecu_1$, and similarly all other $\vecu_i$ using $f_i$.

Our approach is thus to find such polynomials $f_i$.
What we will do is set up a system of linear equations in the coefficients of an unknown degree $r$ polynomial $f$ and show that $f_1$ is the unique solution to the system.
Indeed, showing that $f_1$ is a solution is easy and the hard part is proving that it is the unique solution.

To explain how we set the system of equations, let us assume for the time being that we actually know $\vecu_1$.
Let $f = \sum_{M\in \Monset(m,r)} c_M \cdot M$, where we think of $\{c_M\}$ as unknowns.
Consider the following linear system:
\begin{enumerate}
\item $\sum\limits_{i=1}^t f(\vecu_i) \sspaced{=} f(\vecu_1) \sspaced{=} 1$,
\item $\sum\limits_{i=1}^t (f \cdot M)(\vecu_i) \sspaced{=} M(\vecu_1)$,  for all $M \in \Monset(m,r)$.
\item $\sum\limits_{i=1}^t (f \cdot M \cdot (x_\ell + (\vecu_1)_\ell + 1))(\vecu_i) \sspaced{=} M(\vecu_1)$ for every $\ell \in [m]$ and for all $M \in \Monset(m,r)$.
\end{enumerate}
In words, we have a system of $2 + \binom{m}{\le r} + m\cdot \binom{m}{\le r}$ equations in $\binom{m}{\le r}$ variables (the coefficients of $f$).
Observe that $f=f_1$ is indeed a solution to the system.
To prove that it is the unique solution we rely on the fact that the columns of $U^r$ are linearly independent and hence expressing $\vecu_1^r$ as a linear combination of those columns can be done in a unique way.

Now we explain what to do when we do not know $\vecu_1$.
Let $\vecv=(v_1,\ldots,v_m)\in \F_2^m$.
We modify the linear system above to:
\begin{enumerate}
\item $\sum\limits_{i=1}^t f(\vecu_i) \sspaced{=} f(\vecv) \sspaced{=} 1$,
\item $\sum\limits_{i=1}^t (f \cdot M)(\vecu_i) \sspaced{=} M(\vecv)$ for all $M \in \Monset(m,r)$. \label{item:in-span}
\item $\sum\limits_{i=1}^t (f \cdot M \cdot (x_\ell + v_\ell + 1))(\vecu_i) \sspaced{=} M(\vecv)$ for all $\ell \in [m]$ and $M \in \Monset(m,r)$. \label{item:in-set}
\end{enumerate}
Now the point is that one can prove that if a solution exists then it must be the case that $\vecv$ is an element of $U$.
Indeed, the set of equations in \autoref{item:in-span} implies that $\vecv^r$ is in the linear span of the columns of $U^r$.
The linear equations in \autoref{item:in-set} then imply that $\vecv$ must actually be in the set $U$.

Notice that what we actually do amounts to setting, for every $\vecv\in\F_2^m$, a system of linear equations of size roughly $\binom{m}{\le r}$.
Such a system can be solved in time $\poly\inparen{\binom{m}{\le r}}$.
Thus, when we go over all $\vecv\in\F_2^m$ we get a running time of $2^m \cdot \poly\inparen{\binom{m}{\le r}}$, as claimed.

Our proof can be viewed as an algorithmic version of the proof of Theorem 1.8 of Abbe et al.\ \cite{AbbeSW15}.
That theorem asserts that when the columns of $U^r$ are linearly independent, the $(2r+1)$-syndrome of $U$ is unique.
In their proof of the theorem they first use the $(2r)$-syndrome to claim that if $V$ is another set with the same $(2r)$-syndrome then the column span of $U^r$ is the same as that of $V^r$.
Then, using the degree $(2r+1)$ monomials they deduce that $U=V$.
This is similar to what our linear system does, but, in contrast, \cite{AbbeSW15} did not have an efficient algorithmic version of this statement.

\section{Decoding Algorithm For Reed-Muller Codes}\label{sec:decode}

We begin with the following basic linear algebraic fact.
\begin{lemma}
\label{lem:dual-poly}
Let $\vecu_1,\dots, \vecu_t \in \F_2^m$ such that $\inbrace{\vecu_1^{r}, \dots, \vecu_t^{ r}}$ are linearly independent.
Then, for every $i \in [t]$, there exists a polynomial $f_i$ so that for every $j \in [t]$,
\[
f_i (\vecu_j) = \delta_{i,j} = \begin{cases}
1 & \text{if } i=j \\
0 & \text{otherwise}.
\end{cases}
\]
\end{lemma}
\noindent For completeness, we give the short proof.
\begin{proof}
  Consider the matrix $U^r \in \F_2^{t \times \binom{m}{\le r}}$ whose $i$-th row is $\vecu_i^{ r}$.
A polynomial $f_i$ which satisfies the properties of the lemma is a solution to the linear system $U^r\vecx=\vece_i$, where $\vece_i \in \F_2^t$ is the $i$-th elementary basis vector (that is, $(\vece_i)_j=\delta_{i,j}$), and the $\binom{m}{\le r}$ unknowns are the coefficients of $f_i$.
By the assumption that $U$ is of full rank, indeed there exists a solution.
\end{proof}

\noindent 
The algorithm would proceed by making a guess $\vecv=(v_1,\dots, v_m) \in \F_2^m$ for one of the error locations.
If we could come up with an efficient way to \emph{verify} that the guess is correct, this would immediately yield a decoding algorithm.
We shall verify our guess by using the dual polynomials $f_1,\dots, f_t$ described above.
We shall find them by solving a system of linear equations that can be constructed from the $(2r+1)$-syndrome of $\inbrace{\vecu_1,\dots, \vecu_m}$.
We will need the following crucial, yet simple, observation.

\begin{observation}\label{obs:evalsum-from-syndrome}
  Let $f$ be any $m$-variate polynomial of degree at most $2r+1$, and $\vecu_1,\dots, \vecu_t \in \F_2^m$.
Then, the sum $\sum_{i=1}^t f(\vecu_i)$ can be computed given the $(2r+1)$-syndrome of $\inbrace{\vecu_1,\ldots,\vecu_t}$, in time $O\inparen{\binom{m}{2r+1}}$.
\end{observation}
\begin{proof}
  For any $M \in \Monset(m,2r+1)$, denote $\alpha_M = \sum_{i=1}^t M(\vecu_i)$ (so that $\alpha = (\alpha_M)_{M \in \Monset(M,2r+1)}$ is precisely the syndrome of $\inbrace{\vecu_1,\ldots,\vecu_t}$).
Write $f = \sum_{M \in \Monset(m,2r+1)} c_M \cdot M$, where $c_M \in \F_2$, then
\begin{align*}
\sum_{i=1}^t f(\vecu_i) &\spaced{=} \sum_{i=1}^t  \sum_{M \in \Monset(m,2r+1)} c_M \cdot M(\vecu_i)\\& \spaced{=} \sum_{M \in \Monset(m,2r+1)} c_M \inparen{\sum_{i=1}^t M(\vecu_i)} \spaced{=} \sum_{M \in \Monset(m,2r+1)} c_M \alpha_M. \qedhere
\end{align*}

\end{proof}

The following lemma shows how to verify a guess for an error location.
It is the main ingredient in the analysis of our algorithm and the reason why it works.
Basically, the lemma gives a system of linear equations whose solution enables us to decide whether a given $\vecv\in \F_2^m$ is a corrupted coordinate or not, without knowledge of the set of errors $U$ but only of its syndrome.
In a sense, this lemma is analogous to the Berlekamp-Welch algorithm, which also gives a system of linear equations whose solution reveals the set of erroneous locations (\cite{BerWel}, and see also the exposition in Chapter 13 of \cite{GRSBook}).

\begin{lemma}[Main Lemma]
\label{lem:decode}
Let $\vecu_1,\dots, \vecu_t \in \F_2^m$ such that $\inbrace{\vecu_1^{r}, \dots, \vecu_t^{ r}}$ are linearly independent, and $\vecv=(v_1,\ldots,v_m) \in \F_2^m$.
Suppose there exists a multilinear polynomial $f \in \F_2[x_1,\ldots,x_m]$ with $\deg(f) \le r$ such that for every monomial $M \in \Monset(m,r)$,
\begin{enumerate}
\item \label{item:sum1} $\sum\limits_{i=1}^t f(\vecu_i) \sspaced{=} f(\vecv) \sspaced{=} 1$,
\item \label{item:2r} $\sum\limits_{i=1}^t (f \cdot M)(\vecu_i) \sspaced{=} M(\vecv)$, and
\item \label{item:2r+1} $\sum\limits_{i=1}^t (f \cdot M \cdot (x_\ell + v_\ell + 1))(\vecu_i) \sspaced{=} M(\vecv)$ for every $\ell \in [m]$.
\end{enumerate}
Then there exists $i \in [t]$ such that $\vecv=\vecu_i$.

\end{lemma}

\noindent Observe that if indeed $\vecv=\vecu_i$ for some $i \in [t]$,
then the polynomial $f_i$ guaranteed by \autoref{lem:dual-poly} satisfies those equations.
Hence, the lemma should be interpreted as saying the converse: that if there exists such a solution, then $\vecv=\vecu_i$ for some $i$.
Further, given the $(2r+1)$-syndrome of $\inbrace{\vecu_1,\ldots,\vecu_t}$ as input, \autoref{obs:evalsum-from-syndrome} shows that each of the above constraints are linear constraints in the coefficients of $f$.
Thus, finding such an $f$ is merely solving a system of $O\inparen{\binom{m}{\leq r}}$ linear equations in $\binom{m}{\leq r}$ unknowns and can be done in $\poly\inparen{\binom{m}{\leq r}}$ time.

\begin{proof}[Proof of \autoref{lem:decode}]
Let $J = \inbrace{j \mid f(\vecu_j) = 1}$. Note that by \autoref{item:sum1} it holds that $J \neq \emptyset$. 

\begin{addmargin}[2em]{5em}
  \begin{subclaim}
    $\sum\limits_{i \in J} \vecu_i^{ r} =\vecv^{ r}$.
  \end{subclaim}
  \begin{myproof}{Subclaim}
    Let $M \in \Monset(m,r)$.
We show that $\sum_{i \in J} M(\vecu_i) = M(\vecv)$, i.e.,\ that the $M$'th coordinate of $\sum_{i \in J} \vecu_i^{ r}$ is equal to that of $\vecv^{ r}$.
Indeed, as $f$ satisfies the constraints in \autoref{item:2r},
    \[
    M(\vecv) = \sum_{i=1}^t (f \cdot M)(\vecu_i) = \sum_{i \in J} (f \cdot M)(\vecu_i) + \sum_{i \not\in J} (f \cdot M)(\vecu_i)
    = \sum_{i \in J} M(\vecu_i). \myqedhere
    \] 

  \end{myproof}
\end{addmargin}
For any $\ell \in [m]$, let $J_{\ell} = \inbrace{j \mid f(\vecu_j)=1 \;\text{and}\; (\vecu_j)_\ell = v_\ell} \subseteq J$.
Observe that this definition implies that for every $j \in [t]$, the index $j$ is in $J_\ell$ if and only if $(f\cdot (x_\ell + v_\ell + 1))(\vecu_j)=1$.
Using a similar argument, we can show the following.

\begin{addmargin}[2em]{5em}
  \begin{subclaim}
    For every $\ell \in [m]$, 
    \begin{equation}
      \label{eq:sum-J_ell}
      \sum_{i \in J_\ell} \vecu_i^{ r} = \vecv^{ r}.
    \end{equation}
  \end{subclaim}
  \begin{myproof}{Subclaim}
    Again, for any $M \in \Monset(m,r)$ the constraints in \autoref{item:2r+1} imply that
    \[
    M(v) = \sum_{i=1}^t (f \cdot M \cdot (x_\ell + v_\ell + 1))(\vecu_i) = \sum_{i \in J_\ell} M(\vecu_i). \myqedhere
    \]
  \end{myproof}
\end{addmargin}

\noindent From the above claims,
\[
\vecv^r = \sum_{i \in J} \vecu_i^r = \sum_{i\in J_1} \vecu_i^r = \dots = \sum_{i \in J_m} \vecu_i^r.
\]
By the linear independence of $\inbrace{\vecu_1^{ r}, \dots, \vecu_t^{r}}$, it follows that $J = J_1 = J_2 = \cdots = J_m$.
Indeed, there is a unique linear combination of $\inbrace{\vecu_1^r,\ldots,\vecu_t^r}$ that gives $\vecv^r$.
The only vector which can be in the (non-empty) intersection $\bigcap_{k=1}^m J_k$ is $\vecv$, and so there exists $i \in [t]$ so that $\vecu_i = \vecv$.
\end{proof}

\autoref{lem:decode} implies a natural algorithm for decoding from $t$ errors indexed by vectors $\inbrace{\vecu_1,\ldots,\vecu_t}$, assuming $\inbrace{\vecu_1^{ r}, \dots, \vecu_t^{ r}}$ are linearly independent, that we write down explicitly in \autoref{alg:decoding}.
\begin{algorithm}
  \caption{: Reed-Muller Decoding}
  \label{alg:decoding}
\begin{algorithmic}[1]
  \Require{A $(2r+1)$-syndrome of $\inbrace{\vecu_1,\ldots,\vecu_t}$}
  \State{$\mathcal{E} = \emptyset$}
  \ForAll{$\vecv=(v_1,\dots, v_m) \in \F_2^m$}
  \State{Solve for a polynomial $f \in \F_2[x_1,\dots, x_m]$ of degree at most $r$: \begin{itemize}
    \item $\sum\limits_{i=1}^t f(\vecu_i)=f(\vecv) = 1$,
    \item $\sum\limits_{i=1}^t (f \cdot M)(\vecu_i) = M(\vecv)$ for all $M \in \Monset(m,r)$. 
    \item $\sum\limits_{i=1}^t (f \cdot M \cdot (x_\ell + v_\ell + 1))(\vecu_i) = M(\vecv)$ for all $\ell \in [m]$ and $M \in \Monset(m,r)$. 
    \end{itemize}}
  \If{there is a polynomial $f$ that satisfies the above system of equations}
  \State{Add $\vecv$ to the set $\mathcal{E}$. }
  \EndIf
  \EndFor
  \State{{\bf return} the set $\mathcal{E}$ as the error locations. }
\end{algorithmic}
\end{algorithm}

\begin{theorem}
\sloppy
\label{thm:decode-algo}
Given the $(2r+1)$-syndrome of $t$ unknown vectors $\inbrace{\vecu_1,\ldots,\vecu_t} \subseteq \F_2^m$ such that $\inbrace{\vecu_1^{ r}, \dots, \vecu_t^{ r}}$ are linearly independent, \autoref{alg:decoding} outputs $\inbrace{\vecu_1,\ldots,\vecu_t}$, runs in time $2^m \cdot \poly(\binom{m}{\le r})$ and can be realized using a circuit of depth $\poly(m) = \poly(\log n)$.
\end{theorem}

\begin{proof}
  The algorithm enumerates all vectors in $\F_2^m$, and for each candidate $\vecv$ checks whether there exists a solution to the linear system of $\poly(\binom{m}{\le r})$ equations in $\poly(\binom{m}{\le r})$ unknowns given in \autoref{lem:decode}.
\autoref{obs:evalsum-from-syndrome} shows that this system of linear equations can be constructed from the $(2r+1)$-syndrome in $\poly(\binom{m}{\leq r})$ time.

By \autoref{lem:dual-poly} and \autoref{lem:decode}, a solution to this system exists if and only if there is $i \in [t]$ so that $\vecv=\vecu_i$.
The bound on the running time follows from the description of the algorithm.
Furthermore, all $2^m=n$ linear systems can be solved in parallel, and each linear system can be solved with an $\mathsf{NC}^2$ circuit (see, e.g., \cite{MahajanV97}).
\end{proof}

Observe that the the proof of correctness for \autoref{alg:decoding} is valid, for any value of $r$, whenever the set of error locations $\inbrace{\vecu_1,\ldots,\vecu_t}$ satisfies the property that $\inbrace{\vecu_1^r,\ldots,\vecu_t^r}$ are linearly independent.
Therefore, we would like to apply \autoref{thm:decode-algo} in settings where $\inbrace{\vecu_1,\ldots,\vecu_t}$ are linearly independent with high probability.

For the constant rate regime, Kumar and Pfister \cite{KumarPfister15} and Kudekar, Mondelli, \c{S}a\c{s}o\u{g}lu and Urbanke \cite{KudekarMSU15}  proved that $RM(m, m-r-1)$ achieves capacity for $r=m/2 \pm O(\sqrt{m})$.

\begin{theorem}[\cite{KumarPfister15}, Theorem 23]
\label{thm:capacity-BEC-const-rate}
Let $r \le m$ be integers such that $r=m/2\pm O(\sqrt{m})$.
Then, for $t=(1-o(1))\binom{m}{\le r}$, with probability $1-o(1)$, for a set of vectors $\inbrace{\vecu_1,\ldots,\vecu_t} \subseteq \F_2^m$ chosen uniformly at random, it holds that $\inbrace{\vecu_1^{ r},\ldots,\vecu_t^{ r}}$ are linearly independent over $\F_2^{\binom{m}{\le r}}$.
\end{theorem}

Letting $r=m/2-o(\sqrt{m})$ and looking at the code $RM(m, m-2r-2) = RM(m, o(\sqrt{m}))$ so that $\binom{m}{\le r} = (1/2 - o(1))2^m$, we get the following statement, stated earlier as \autoref{thm:main-low-degree}.

\begin{corollary}
  There exists a (deterministic) algorithm that is able to correct $t = (1/2 - o(1))2^m$ random errors in $RM(m, o(\sqrt{m})$ with probability $1-o(1)$.
The algorithm runs in time $2^m \cdot \inparen{\binom{m}{m/2-o(\sqrt{m}}}^3 \le n^4$.
\end{corollary}

Alternatively, we can pick $r=m/2-O(\sqrt{m})$ and correct $c \cdot 2^m$ random errors in the code $RM(m, O(\sqrt{m}))$, where $c$ is some positive constant that goes to zero as the constant hidden under the big $O$ increases.

For the high-rate regime, recall the following capacity achieving result proved in \cite{AbbeSW15}:
\begin{theorem}[\cite{AbbeSW15}, Theorem 4.5]
\label{thm:random-linear-independent}
\sloppy
Let $\epsilon>0$, $r \le m$ be two positive integers and $t < \binom{m-\log(\binom{m}{\le r})-\log(1/\epsilon)}{\le r}$.
Then, with probability at least $1-\epsilon$, for a set of vectors $\inbrace{\vecu_1,\ldots,\vecu_t} \subseteq \F_2^m$ chosen uniformly at random, it holds that $\inbrace{\vecu_1^{ r},\ldots,\vecu_t^{ r}}$ are linearly independent over $\F_2^{\binom{m}{\le r}}$.
\end{theorem}

Using \autoref{thm:random-linear-independent}, we apply \autoref{thm:decode-algo} to obtain the following corollary, which was stated informally as \autoref{thm:main:informal}.

\begin{corollary}
  Let $\epsilon>0$, and $r \le m$ be two positive integers.
Then there exists a (deterministic) algorithm that is able to correct $t = \left\lfloor \binom{m-\log(\binom{m}{\le r})-\log(1/\epsilon)}{\le r} \right \rfloor-1$ random errors in $RM(m, m-(2r+2))$ with probability at least $1-\epsilon$.
The algorithm runs in time $2^m \cdot \poly\inparen{\binom{m}{\le r}}$.
\end{corollary}
\noindent If $r=o(\sqrt{m/\log m})$, the bound on $t$ is $(1-o(1))\binom{m}{\le r}$, as promised.

More generally, a positive answer to \autoref{Q:RM:BEC} is equivalent to $\inbrace{\vecu_1^r, \ldots, \vecu_t^r}$ for $t=(1-o(1))\binom{m}{\le r}$ being linearly independent with probability $1-o(1)$ (see Corollary 2.9 in \cite{AbbeSW15}), and thus we also obtain the following corollary, which was stated informally as \autoref{thm:main:conj:informal}.

\begin{corollary}
  Let $r \le m$ be two positive integers.
Suppose that $RM(m, m-r-1)$ achieves capacity for the BEC.
Then there exists a (deterministic) algorithm that is able to correct $(1-o(1))\binom{m}{\le r}$ random errors in $RM(m, m-(2r+2))$ with probability $1-o(1)$.
The algorithm runs in time $2^m \cdot \poly\inparen{\binom{m}{\le r}}$.
\end{corollary}

We note that for all values of $r$, $2^m \cdot \poly\inparen{\binom{m}{\le r}}$ is polynomial in the block length $n=2^m$, and when $r=o(m)$ this is equal to $n^{1+o(1)}$.

\section{Abstractions and Generalizations}\label{sec:abstraction}

\subsection{An abstract view of the decoding algorithm}

In this section we present a more abstract view of \autoref{alg:decoding}, in the spirit of the works by Pellikaan, Duursma and K\"{o}tter (\cite{Pellikaan92, DuursmaK94}) which abstract the Berlekamp-Welch algorithm (see also the exposition in \cite{SudanLectureNotes}).
Stated in this way, it is also clear that the algorithm works also over larger alphabets, so we no longer limit ourselves to dealing with binary alphabets.
As shown in \cite{KumarPfister15}, Reed-Muller codes over $\F_q$ (sometimes referred to as {\em Generalized Reed-Muller codes)} also achieve capacity in the constant rate regime.

We begin by giving the definition of a (pointwise) product of two vectors, and of two codes.

\begin{defin}
\label{def:star}
Let $\vecu,\vecv \in \F_q^n$.
Denote by $\vecu * \vecv \in \F_q^n$ the vector $(\vecu_1\vecv_1, \ldots, \vecu_n \vecv_n)$.
For $A,B \subseteq \F_q^n$ we similarly define $A*B = \setdef{\vecu * \vecv}{\vecu \in A, \vecv \in B}$.
\end{defin}

Following the footsteps of \autoref{alg:decoding}, we wish to decode, in a code $C$, error patterns which are correctable from erasures in a related code $N$, through the use of an {\em error-locating code} $E$.
Under some assumptions on $C, N$ and $E$, we can use a similar proof in order to do this.

\begin{theorem}
\label{thm:abstraction}
Let $E, C, N \subseteq \F_q^n$ be codes with the following properties.
\begin{enumerate}
\item \label{item:inclusion} $E*C \subseteq N$
\item \label{item:correction} For any pattern $\1_{U}$ that is
  correctable from erasures in $N$, and for any coordinate $i \not \in U$ there exists a codeword $\vece \in E$ such that $\vece_j = 0$ for all $j \in U$ and $\vece_i=1$.
\end{enumerate}
Then there exists an efficient algorithm that corrects in $C$ any pattern $\1_U$, which is correctable from erasures in $N$.\end{theorem}

To put things in perspective, earlier we set $C=RM(m, m-2r-2)$, $N=RM(m, m-r-1)$ and $E=RM(m,r+1)$.
It is immediate to observe that \autoref{item:inclusion} holds in this case, and \autoref{item:correction} is guaranteed by \autoref{lem:dual-poly}: Indeed, consider the error pattern $U=\{\vecu_1, \ldots, \vecu_t\}$ and the dual polynomials $\{f_i\}_{i=1}^t$, and let $\vecv \not\in U$ be any other coordinate of the code.
If there exists $j \in [t]$ such that $f_j (\vecv)=1$, we can pick the codeword $g=f_j \cdot (1+x_\ell + \vecv_\ell)$, where $\ell$ is some coordinate such that $\vecv_\ell \neq (\vecu_j)_\ell$.
$g$ has degree at most $r+1$ and so it is a codeword in $E$, and it can be directly verified that it satisfies the conditions of \autoref{item:correction}.
If $f_j(\vecv)=0$ for all $j$, we can pick $g=1-\sum_{i=1}^t f_i$.

It is also worth pointing out the differences between our approach and the abstract Berlekamp-Welch decoder of Duursma and K\"{o}tter: They similarly set up codes $E, C$ and $N$ such that $E*C \subseteq N$.
However, instead of \autoref{item:correction}, they require that for any $\vece \in E$ and $\vecc \in C$, if $\vece * \vecc = 0$ then $\vece=0$ or $\vecc=0$ (or similar requirements regarding the distances of $E$ and $C$ that guarantee this property).
This property, as well as the distance properties, do not hold in the case of Reed-Muller codes.

Turning back to the proof of \autoref{thm:abstraction}, the algorithm and the proof of correctness turn out to be very short to describe in this level of generality.
Given a word $\vecy \in \F_q^n$, the algorithm would solve the the linear system $\veca * \vecy = \vecb$, in unknowns $\veca \in E$ and $\vecb \in N$.
Under the hypothesis of the theorem, we show that common zeros of the possible solutions for $\veca$ determine exactly the error locations.
Once the locations of the errors are identified, correcting them is easy: we can replace the error locations by the symbol '?'
and use an algorithm which corrects erasures (this can always be done efficiently, when unique decoding is possible, as this merely amounts to solving a system of linear equations).
The algorithm is given in \autoref{alg:abstract}.

\begin{algorithm}
  \caption{: Abstract Decoding Algorithm}
  \label{alg:abstract}
\begin{algorithmic}[1]
  \Require{received word $\vecy \in \F_q^n$ such that $\vecy = \vecc + \vece$, with $\vecc \in C$ and $\vece$ is supported on a set $U$}
  \State{Solve for $\veca \in E, \vecb \in N$, the linear system $\veca * \vecy = \vecb$.}
  \State{Let $\inbrace{\veca_1, \ldots, \veca_k}$ be a basis for the solution space of $\veca$, and let $\mathcal{E}$ denote the common zeros of $\setdef{\veca_i}{i \in [k]}$.}
  \State{For every $j \in \mathcal{E}$, replace $\vecy_j$ with '?', to get a new word $\vecy'$.}
  \State{Correct $\vecy'$ from erasures in $C$.}
\end{algorithmic}
\end{algorithm}

Note that in \autoref{thm:abstraction} we assume that the error pattern $U$ is correctable from erasures in $N$, whereas \autoref{alg:abstract} first computes a set of error locations $\mathcal{E}$ and then corrects $\vecy'$ from erasures in $C$.
Thus, the proof of \autoref{thm:abstraction} can be divided into two steps.
The first, and the main one, will be to show that $\mathcal{E}=U$.
The second, which is merely an immediate observation, will be to show that $U$ is also correctable from erasures in $C$.
We begin with the second part:

\begin{lemma}
\label{lem:correct-from-erasures}
Assume the setup of \autoref{thm:abstraction}, and let $U$ be any pattern which is correctable from erasures in $N$.
Then $U$ is also correctable from erasures in $C$.
\end{lemma}
\begin{proof}
  We may assume that $U \neq \emptyset$, as otherwise the statement is trivial.
Suppose on the contrary that $U$ is not correctable from erasures in $C$, that is, there exists a non-zero codeword $\vecc \in C$ supported on $U$.
For any $\veca \in E$, we have that $\veca * \vecc$ is a codeword of $N$ which is supported on a subset of $U$.
In order to reach a contradiction, we want to pick $\veca \in E$ so that $\veca * \vecc$ is a non-zero codeword of $N$, which contradicts the assumption that $U$ is correctable from erasures in $N$.

Pick $i \in U$ so that $\vecc_i \neq 0$.
Observe that if $U$ is correctable from erasures in $N$ then so is $U \setminus \inbrace{i}$.
By \autoref{item:correction} in \autoref{thm:abstraction} with respect to the set $U \setminus \inbrace{i}$ there exists $\veca \in E$ with $\veca_i=1$.
Thus, in particular $\veca * \vecc$ is non-zero.
\end{proof}

We now prove that main part of \autoref{thm:abstraction}, that is, that under the assumptions stated in the theorem, \autoref{alg:abstract} correctly decodes (in $C$) any error pattern that is correctable from erasures in $N$.

\begin{proof}[Proof of \autoref{thm:abstraction}]
  Write $\vecy = \vecc + \vece$, so that $\vecc \in C$ is the transmitted codeword and $\vece$ is supported on the set of error locations $U$.
As noted above, by \autoref{lem:correct-from-erasures} it is enough to show that under the assumptions of the theorem (in particular, that $U$ is correctable from erasures in $N$), the set of error locations $\mathcal{E}$ computed by \autoref{alg:abstract} equals $U$.

In the following two lemmas, we argue that any solution $\veca$ for the system vanishes on the error points, and then that for every other index $i$, there exists a solution whose $i$-th entry is non-zero (and so there must be a basis element for the solution space whose $i$-th entry is non-zero).

The following lemma states that every solution $\veca \in E$ to the equation $\veca * \vecy = \vecb$ vanishes on $U$, the support of $\vece$.
In the pointwise product notation, this is equivalent to showing that $\veca * \vece = 0$.

\begin{addmargin}[2em]{5em}
  \begin{subclaim}
    For every $\veca \in E, \vecb \in N$ such that $\veca * \vecy = \vecb$, it holds that $\veca * \vece = 0$.
  \end{subclaim}
  \begin{myproof}{Subclaim}
    Since $\veca * \vecy = \vecb \in N$ (by the assumption) and $\veca * \vecc \in N$ (by \autoref{item:inclusion}), we get that $\veca*\vece = \veca * \vecy - \veca * \vecc$ is also a codeword in $N$.
Furthermore, $\veca*\vece$ is also supported on $U$, and since $U$ is an erasure-correctable pattern in $N$, the only codeword that is supported on $U$ is the zero codeword.
  \end{myproof}
\end{addmargin}

\noindent To finish the proof, we show that for any $i \not\in U$,
there is a solution $\veca$ to the system of linear equations with $\veca_i = 1$.

\begin{addmargin}[2em]{5em}
  \begin{subclaim}
    For every $i \not \in U$ there exists $\veca \in E, \vecb \in N$ such that $\veca$ is 0 on $U$, $\veca_i=1$ and $\veca*\vecy=\vecb$.
  \end{subclaim}
  \begin{myproof}{Subclaim}
    By \autoref{item:correction}, since $U$ is correctable from erasures in $N$, for every $i \not \in U$ we can pick $\veca \in E$ such that $\veca$ is 0 on $U$ and $\veca_i = 1$.
Set $\vecb = \veca * \vecy$.
It remains to be shown that $\vecb$ is a codeword of $N$.
This follows from the fact that
	\[
	\vecb = \veca * \vecc + \veca * \vece = \veca * \vecc,
	\]
	where the second equality follows from the fact that $\veca$ is zero on $U$ (the support of $\vece$). Finally, $\veca * \vecc$ is a codeword of $N$ by \autoref{item:inclusion}.
  \end{myproof}
\end{addmargin}
These two claims complete the proof of the theorem.
\end{proof}

\subsection{Decoding of Linear Codes over $\F_2$}
\label{sec:general}

In \cite{AbbeSW15}, it is observed that their results for Reed-Muller codes imply that for {\em every} linear code $N$, every pattern which is correctable from erasures in $N$ is correctable from errors in what they call the ``degree-three tensoring'' of $N$.
One can in fact use our \autoref{alg:decoding} almost verbatim to obtain an efficient version of this statement.
However, here we remark that this is nothing but a special case of \autoref{thm:abstraction} with an appropriate setting of the codes $E,C,N$.
We begin by briefly describing their definitions and their argument.

The basic tool used by \cite{AbbeSW15} is embedding any parity check matrix in the matrix $E(m,1)$ for an appropriate choice of $m$.
Let $N$ be any linear code of dimension $k$ over $\F_2$ and $H$ be its parity check matrix.
For convenience, we first extend $N$ by adding a parity bit.
This increases the block length by 1, does not decrease the distance and preserves the dimension.
A parity check matrix for the extended code can by obtained from $H$ by constructing the matrix

\[
H_0=
\left(
\begin{array}{cc}
  1 & 1 \cdots 1 \\ 
  0 & \raisebox{-15pt}{{\huge\mbox{{$H$}}}} \\[-3ex]
  \vdots & \\
  0 &
\end{array}
\right).
\]

The main observation now is that $E(m,1)$ is an $(m+1)\times 2^m$ matrix that contains all vectors of the form $(1,\vecv)$ for $\vecv \in \F_2^m$, so if we set $m=n-k$ to be the number of rows of $H$, we can pick a subset $S$ of the columns of $E(m,1)$ that correspond to the columns that appear in $H_0$.

\cite{AbbeSW15} then define the degree-three tensoring of $N$, which is a code $C$ whose parity check matrix is $H_0^{\otimes 3}$: this is an $\binom{m}{\le 3} \times n$ matrix with rows indexed by tuples $i_1 < i_2 < i_3$, with the corresponding row being the pointwise product (as in \autoref{def:star}) of rows $i_1,i_2,i_3$ of $H_0$.
One can then verify that \autoref{alg:decoding} can be used in order to correct (in $C$) any error pattern which is correctable from erasures in $N$, by using the algorithm with $r=1$ and having the error location guesses run only over the columns in $S$.

A closer look reveals that this construction is in fact a special case of \autoref{thm:abstraction}.
Given any linear binary code $N$ with parity check matrix $H$, the main observation of \cite{AbbeSW15} can be interpreted as saying that when we add a parity bit to $N$, we can embed $N$ in a puncturing of $RM(m, m-2)$ (whose parity check matrix is $E(m,1)$).
We state it in the following claim:

\begin{claim}
\label{cl:embed}
Let $N'$ denote the subcode of $RM(m, m-2)$ of all words that are 0 outside $S$. Then $N$ is precisely the restriction of $N'$ to the $S$ coordinates.
\end{claim}

\begin{proof}
  Let $\vecb \in N$.
Then $H_0 \vecb=0$, i.e.\ the columns of $H_0$ indexed by the non-zero elements in $\vecb$ add up to 0.
Let $\vecb' \in \F_2^{2^m}$ denote that extension of $\vecb$ into a vector of length $2^m$ obtained by filling 0's in every coordinate not in $S$.
Then $E(m,1)\vecb'=0$, since the same columns that appeared in $H_0$ appear in $E(m,1)$.
This implies that $\vecb' \in N'$.

Similarly, for every $\vecb' \in N'$, we can define $\vecb$ to be its restriction to $S$, and then $H_0\vecb=0$, i.e.\ $\vecb \in N$.
\end{proof}

The degree-three tensoring of $N$, which we denote by $C$, can then be similarly embedded in a puncturing of $RM(m, m-4)$, where again, only the coordinates in $S$ remain, and similarly $C$ can be seen to be the restriction to $S$ to the subcode $C'$ of $RM(m, m-4)$ that contains the words that are 0 outside $S$.

Finally, we define the error locating code $E$ to be the restriction of $RM(m,2)$ to the coordinates of $S$.

We now show that the conditions of \autoref{thm:abstraction} are satisfied in this case.
We begin with \autoref{item:correction}.
If $U$ is a correctable pattern in $N$, it means that the columns indexed by $U$ in $H_0$ are linearly independent.
It follows that they are also linearly independent as columns in $E(m,1)$.
Hence, using the same arguments as before we can find, for any coordinate $\vecv \not \in U$, a degree 2 polynomial $g$ such that $g(\vecv)=1$ and $g$ restricted to $U$ is 0.
Restricting the evaluations of $g$ to the subset of coordinates $S$, we get a codeword $\vece \in E$ with the required property.

As for \autoref{item:inclusion}: We first argue that $RM(m,2)*C' \subseteq N'$, since the degrees match and the property of vanishing outside $S$ is preserved under multiplication.
Projecting back to the coordinates in $S$, we get that $E*C \subseteq N$.

\section*{Acknowledgement}
We would like thank Avi Wigderson, Emmanuel Abbe and Ilya Dumer for helpful discussions and for commenting on an earlier version of the paper. We thank Venkatesan Guruswami and anonymous reviewers for pointing out the abstraction of \autoref{alg:decoding} given in \autoref{sec:abstraction}.

\bibliographystyle{customurlbst/alphaurlpp}
\bibliography{RMcode}

\appendix

\end{document}